\documentclass[11pt]{article}

\usepackage{amsmath,amssymb,amsfonts,amsthm,bm,thmtools}
\usepackage{hyperref}
\usepackage[nameinlink]{cleveref}
\usepackage[textsize=scriptsize,shadow]{todonotes}

\hypersetup{
	colorlinks=true,
	linkcolor=blue
}

\theoremstyle{theorem}
\declaretheorem[name=Theorem]{theorem}
\declaretheorem[name=Lemma,sibling=theorem]{lemma}
\declaretheorem[name=Proposition,sibling=theorem]{proposition}
\declaretheorem[name=Corollary,sibling=theorem]{corollary}

\theoremstyle{definition}
\declaretheorem[name=Definition]{definition}
\declaretheorem[name=Assumption]{assumption}

\newcommand{\vocab}[1]{\textbf{#1}}
\newcommand{\numeraire}{num\'{e}raire}

\newcommand{\E}{\mathbb{E}}

\newcommand{\R}{\mathbb{R}}

\newcommand{\cS}{\mathcal{S}}

\newcommand{\PNL}{\mathsf{PNL}}
\newcommand{\MEV}{\mathsf{MEV}}

\DeclareMathOperator{\pathmev}{path}

\begin{document}

\title{Invariance properties of maximal extractable value}
\author{
	Alan Guo
	\thanks{Jump Crypto. {\tt aguo@jumptrading.com}}
}
\maketitle

\begin{abstract}
We develop a formalism for reasoning about trading on decentralized exchanges
on blockchains and a formulation of a particular form of maximal extractable value (MEV)
that represents the total arbitrage opportunity extractable from on-chain liquidity.
We use this formalism to prove that for blockchains with deterministic block times
whose liquidity pools satisfy some natural properties that are satisfied by pools in practice,
this form of MEV is invariant under changes to the ordering mechanism of the
blockchain and distribution of block times.
We do this by characterizing the MEV as the profit of a particularly simple
arbitrage strategy when left uncontested.
These results can inform design of blockchain protocols by ruling out
designs aiming to increase trading opportunity by changing the
ordering mechanism or shortening block times.
\end{abstract}

\section{Introduction}

\subsection{Background}
Blockchain systems originated with Bitcoin as simply distributed ledgers, but have
since evolved into distributed universal computers starting with Ethereum.
Such Turing-complete blockchains enable decentralized finance (DeFi) via programs, dubbed ``smart contracts'',
which facilitate financial transactions between parties without intermediaries.

One of the most common types of protocols in decentralized finance are decentralized exchanges.
The dominant model for a decentralized exchange is an automated market maker (AMM).
AMMs are typically smart contracts that have passive liquidity providers add assets to the contract.
Active traders may trade with the AMM by making ``swaps'' against the contract,
trading one asset for another, and pay a fee on each swap to the liquidity providers
to compensate them.
What distinguishes an AMM from any other is the logic it uses to govern what
swaps are valid from any given state of the AMM.
A constant function market maker (CFMM) is a prototypical example of an AMM.
A CFMM is specified by an invariant function of its asset reserves, such that
a swap is only valid if it maintains the invariant.
One of the simplest and still most popular forms of CFMM is constant product market makers,
where the invariant is the product of the asset reserves.
AMMs generally rely on arbitrageurs to bring prices in line with external prices.
Therefore, the available liquidity in an AMM has option value to arbitrageurs.

MEV is a phenomenon inherent to blockchains that support decentralized finance.
The acronym was coined by \cite{Flashboys} and originally stood for ``miner extractable value'',
and is traditionally defined as the value that can be extracted by a miner by reordering,
adding, or censoring transactions from blocks.
It was so named because at the time Ethereum, the blockchain with the first and most mature
DeFi ecosystem, used a proof-of-work (PoW) consensus protocol in which
miners had absolute control over the transactions that were ultimately included in blocks.
Prominent examples of malignant MEV include front running and sandwich attacks, which extract profits at the
expense of other users. However, there are also benign forms of MEV such as cross-exchange arbitrage.

More recently, DeFi-capable blockchains typically do not use PoW for consensus --- even Ethereum has
since switched to a proof-of-stake (PoS) consensus protocol. As such, miners no longer play a role
in the network, and MEV is now referred to as ``maximal extractable value'', which is a less well-defined concept.
What is clear is that users assign value to on-chain execution and are willing to pay fees
to network operators such as validators to include their transactions in blocks.
So there is clearly value that users wish to extract from the network, a portion of which accrues to validators as fees.

An active area of research in the blockchain and DeFi community is designing protocols to
minimize malignant MEV and to maximize benign MEV, the former to encourage adoption and
the latter to increase the value of the network and incentivize network operators.
One challenge in this research is determining whether a proposed design actually increases the overall extractable value
or merely redistributes it. It is difficult to assess the impact of a design if the objective
to optimize is not clearly, precisely defined.

\subsection{Our results}
The goal of this work is to formalize an intuitive notion of benign, arbitrage-centric MEV,
and to prove some invariance properties of this notion of MEV, showing that certain changes to
a blockchain design do not increase the overall extractable value.
We emphatically do not consider the malignant redistributive MEV in this work.

Our work introduces a formalism for reasoning about trading strategies that have access
to a liquid external market, such as a centralized exchange, as well as decentralized
exchanges on blockchains with discrete block times.
At the core is the formalization of liquidity pools as abstract state machines satisfying natural axioms,
which abstracts properties of not only CFMMs, but other AMMs such as Uniswap v3 as well, and even compositions of AMMs.
Within this formalism, we define MEV as the maximum net amount of profit that can be
extracted from the network before gas fees. This quantity represents an upper bound on the amount
traders would be willing to pay in fees to network operators.
We define competitive and noncompetitive versions of MEV.
The former represents the extractable value if at each block traders leave no arbitrage opportunity behind,
while the latter represents the same quantity without this constraint.
We also define a particular simple arbitrage strategy that is a greedy locally optimal strategy.
We show that when the asset price is a martingale,
then the competitive MEV is equal to the expected profit of this simple strategy when given
exclusive access to the liquidity,
while the noncompetitive MEV is equal to the expected profit of a related family of strategies.
Since the expected profit of a strategy left uncontested is independent of the blockchain's
ordering mechanism and distribution of block times, this shows the independence of
the corresponding notions of MEV to these same factors.

In \Cref{sec:formalism}, we lay out our formalism for on-chain trading, including
liquidity pools, blockchain markets, trading strategies, PNL, and MEV.
Our notion of liquidity pool is significantly more abstract and general than a single AMM
so that our results may apply more broadly.
In \Cref{sec:main_results}, we state our main invariance results.
In \Cref{sec:proofs}, we develop some theory around our formalism and leverage it
to prove the main results.
In \Cref{sec:generalizations}, we show how some of our simplifying assumptions can be relaxed
and how our results extend to more general settings.
In \Cref{sec:example_pools}, we show how many decentralized exchanges in the wild
satisfy the axioms of our formalism.

\subsection{Prior work}
MEV was introduced by the seminal paper \cite{Flashboys}.
The work of \cite{KDC22} develops a theoretical framework for analyzing the redistributive notion of MEV,
which is orthogonal to the notion of MEV studied in this work.
The constant product market maker model was formalized as a state transition system in \cite{ZCP18}.
Constant product market makers and CFMMs were analyzed in \cite{AKC19} and \cite{AC20} respectively.
The idea that holding a passive liquidity position is equivalent to a short option position is not new.
\cite{Clark} shows that the payoff of a constant product market can be replicated by a
static portfolio of European call and put options.
\cite{LVR} quantifies the ``loss-versus-rebalancing'' (LVR) of passive liquidity in a CFMM,
which represents the optionality captured by arbitrageurs.

\section{Formalism}
\label{sec:formalism}

In this section we develop the formalism for analyzing on-chain trading and MEV.
We give definitions and state our assumptions. Some of the assumptions are purely
for the sake of simplifying the exposition and do not result in the loss of any generality.
We explore how these assumptions can be relaxed in \Cref{sec:generalizations}.

\begin{assumption}[Single risky asset]
	\label{assum:single_risky_asset}
	There are just two assets: a risky asset $x$ and a \numeraire{} $y$.
	In \Cref{sec:multiple_assets} we show how to generalize our results
	to multiple risky assets.
\end{assumption}

\begin{assumption}[Frictionless external market]
	\label{assum:frictionless_external_market}
	There is a external market where $x$ and $y$ may be exchanged
	frictionlessly (instantly without fees) at a market price $P_t$
	which is assumed to be a semimartingale.
\end{assumption}

\subsection{Liquidity pools}
In an effort to capture as large a class of on-chain markets as we can, we model them as abstractly as possible,
only requiring as many assumptions as are necessary (all of which are natural and satisfied by real world examples)
to prove our main results.
Note that throughout this work, we are assuming that liquidity is static --- no liquidity is
added or removed from a pool.

\begin{definition}[Liquidity pool]
	\label{def:liquidity_pool}
	A \vocab{liquidity pool} is an abstract deterministic state machine, represented by the tuple
	\[
	(\Sigma, A, \{A_s\}_{s \in \Sigma}, \tau, \pi, s_0)
	\]
	where
	\begin{itemize}
		\item%
		$\Sigma$ is the \vocab{state space}, the set of possible states the pool can be in.
		An element of $\Sigma$ is called a \vocab{(pool) state}
		
		\item%
		$A$ is the \vocab{atomic action space}, the set of possible atomic actions one can take interacting with the pool.
		
		\item%
		Given $A$, we define $A^*$ to be the smallest set such that $A \subseteq A^*$ and $a_1,a_2 \in A \implies a_1a_2 \in A^*$.
		In other words, $A^*$ is the set of concatenations of elements of $A$. We call elements of $A^*$ \vocab{actions} or \vocab{trades}
		(we will use them interchangeably).
		
		\item%
		For each pool state $s \in \Sigma$, the set $A_s \subseteq A^*$ is the set of \vocab{admissible actions}
		with respect to $s$.
		
		\item%
		$\tau: S \times A^* \to S$ is the \vocab{transition} function, describing how the state of the pool
		changes under each action.
		
		\item%
		$\pi: A^* \to \R^2$ is the \vocab{payoff} function. If action $a$ is taken, then $\pi(a) = (\Delta x,\Delta y)$
		represents the change in the trader's holdings in the assets $x$ and $y$.
		
		\item%
		$s_0$ is the \vocab{initial pool state} at time $t=0$.
	\end{itemize}
\end{definition}

\begin{assumption}[Liquidity pool axioms]
	\label{assum:liq_pool_axioms}
	We assume a liquidity pool satisfies the following axioms:
	\footnote{The algebraically inclined may notice that these axioms declare $A^*$ to be the free monoid on the set $A$,
	and $\pi: A^* \to \R^2$ is a monoid homomorphism, and $\tau: S \times A^* \to S$ describes a monoid action of $A^*$ on $S$.}
	\begin{itemize}
		\item%
		(Null action): There exists a null action $\bot \in A$, such that
		\begin{enumerate}
			\item%
			$\bot \in \bigcap_{s \in \Sigma} A_s$
			
			\item%
			$\tau(s,\bot) = s$ for all $s \in \Sigma$
			
			\item%
			$\pi(\bot) = (0,0)$
		\end{enumerate}
		
		\item%
		(Composition of actions): For any two actions $a_1,a_2 \in A^*$, their \vocab{composite action} $a_1a_2 \in A$
		satisfies
		\begin{enumerate}
			\item%
			$\pi(a_1a_2) = \pi(a_1) + \pi(a_2)$
			
			\item%
			$a_1 \in A_s, a_2 \in A_{\tau(s,a_1)} \implies a_1a_2 \in A_s$
			
			\item%
			Composition is associative: $(a_1a_2)a_3 = a_1(a_2a_3)$
		\end{enumerate}
		
		\item%
		(Optimal action):
		For any pool state $s \in \Sigma$ and external price $P$ of the asset $x$,
		there exists an atomic $a^*(s,P) \in A_s \cap A$ with $\pi(a^*(s,P)) = (\Delta x,\Delta y)$ that maximizes
		\[
		\Delta x P + \Delta y.
		\]
		Note that $\Delta x P + \Delta y \ge 0$ since we can always choose $\bot$.
	\end{itemize}
\end{assumption}

\begin{definition}[Admissible sequence of actions]
	\label{def:admissible_sequence}
	If $s \in \Sigma$ is a pool state, then a sequence $a_1,\ldots,a_n \in A$ is \vocab{admissible} with respect to $s$
	if $a_1\cdots a_n \in A_s$.
\end{definition}

\begin{assumption}[Single liquidity pool]
	\label{assum:single_liq_pool}
	There is a single liquidity pool with reserves in the risky asset $x$ and \numeraire{} $y$.
	This assumption loses no generality, as shown in \Cref{sec:multiple_liq_pools}.
\end{assumption}

\begin{definition}[No-arbitrage state]
	\label{def:no_arb_state}
	For any external price $P$ of the asset $x$, a \vocab{no-arbitrage state} relative to $P$ is a state
	$s \in \Sigma$ such that for every admissible $a \in A_s$,
	\[
	\pi(a) = (\Delta x,\Delta y) \implies \Delta x P + \Delta y \le 0.
	\]
\end{definition}

\begin{definition}[Frictionless pool]
	\label{def:frictionless}
	A liquidity pool is \vocab{frictionless} if for every external price $P$ there exists a unique
	no-arbitrage state $s^*(P) \in \Sigma$.
\end{definition}

\begin{definition}[Path-independent pool]
	\label{def:path_independent}
	A liquidity pool is \vocab{path-independent} if
	\[
	\tau(s,a) = \tau(s,a') \implies \pi(a) = \pi(a')
	\]
	for all $s \in \Sigma$.
\end{definition}

\begin{definition}[Efficient pool]
	\label{def:efficient}
	A liquidity pool is \vocab{efficient} if every state is a no-arbitrage state relative to some price.
\end{definition}

\begin{definition}[Volume]
	\label{def:volume}
	For any action $a \in A^*$, the \vocab{volume} of $a$, denoted $|a|$, is defined as
	\[
	|a| \triangleq
	\begin{cases}
		|\Delta y| & a \in A ~\text{and}~\pi(a) = (\Delta x,\Delta y) \\
		|a_1| + |a_2| & a = a_1a_2
	\end{cases}
	\]
	which is well-defined since composition in $A^*$ is associative.
\end{definition}

\begin{definition}[Liquidity pool with fees]
	\label{def:liq_pool_fees}
	Let $\phi > 0$. A liquidity pool $\Pi = (\Sigma, A, \{A_s\}_{s \in \Sigma}, \tau, \pi, s_0)$
	\vocab{has fee $\phi$} if there exists a payoff function $\pi_0$ such that the
	liquidity pool $\Pi_0 = (\Sigma, A, \{A_s\}_{s \in \Sigma}, \tau, \pi_0, s_0)$ satisfies
	the axioms of \Cref{assum:liq_pool_axioms} and
	\[
	\pi_0(a) = (\Delta x,\Delta y) \implies \pi(a) = (\Delta x,\Delta y - \phi|a|).
	\]
	We say $\Pi$ has fee $\phi$ \vocab{relative to} $\Pi_0$, which is the \vocab{underlying} pool.
\end{definition}

A pool satisfying \Cref{def:liq_pool_fees} does not necessarily satisfy the axioms of \Cref{assum:liq_pool_axioms},
in particular the optimal action axiom. We show (\Cref{cor:s0_fee_well_defined}) that if the underlying pool is
efficient and frictionless, then the pool with fees satisfies the optimal action axiom.

In \Cref{sec:example_pools} we see that the most popular, dominant decentralized exchanges satisfy the
liquidity pool axioms and furthermore are frictionless, path-independent, and efficient.

\subsection{Blockchain markets}
\begin{definition}[Blockchain market]
	\label{def:blockchain_market}
	A blockchain market is an increasing sequence of \vocab{block times} $t_1,t_2,\ldots \in \R_+$ and an
	\vocab{ordering mechanism} which is a probability distribution for each $t_n$ satisfying
	the following:
	\begin{itemize}
		\item%
		For each block time $t_n$, let $s_n$ be the pool state at time $t_n$ before any actions.
		For each sequence $T$ of actions (not necessarily admissible), the ordering mechanism
		specifies a probability distribution over permutations $\sigma(T)$ of $T$.
		The sequence $T$ is the \vocab{on-chain trades submitted for block $n$}.
		
		\item%
		Once a permutation $\sigma(T)$ is selected, a subsequence $T'$ is constructed from
		$\sigma(T)$ by starting from the first element of $\sigma(T)$ and iteratively dropping
		any action that is not admissible with respect to the pool state after all the preceding non-dropped
		trades. By construction, the resulting subsequence $T'$ is admissible with respect to
		the original pool state $s_n$. The sequence $T'$ is the
		\vocab{on-chain trades executed for block $n$}.
		
		\item%
		The starting state at block time $t_{n+1}$ is $s_{n+1} = \tau(s_n,T')$ where we abuse
		notation by using $T'$ to denote the composition of the elements of $T'$.
	\end{itemize}
\end{definition}

\begin{assumption}[Deterministic block times]
	\label{assum:deterministic_block_times}
	We assume that the block times $(t_n)$ are deterministic.
\end{assumption}

\begin{assumption}[No gas fees]
	\label{assum:no_gas_fees}
	We assume there are no gas fees.
	As such, the PNL we measure for a trading strategy is the gross PNL
	(after liquidity provider fees) and an upper bound on how much the
	strategy would be willing to pay in fees to network operators.
\end{assumption}

\subsection{Trading strategies}
We extend the formalization of a trading strategy used in \cite{LVR} to explicitly reflect the ability
to make discrete trades on a decentralized exchange on a blockchain.

\begin{definition}[Trading strategies, on-chain trades]
	\label{def:trading_strategy}
	A \vocab{trading strategy} is a stochastic process $(x_t,y_t)$ representing its holdings at time $t$
	of the risky asset $x$ and the \numeraire{} $y$ and, for each block time $t_n$, a sequence of
	\vocab{submitted on-chain trades} $T_n$, which may be chosen probabilistically.
	Because we assume block times are deterministic, we may assume that the on-chain trades are submitted
	using all information available at time $t_n$.
	
	We require that the process is adaptable, predictable, and satisfies
	\[
	\E\left[ \int_0^t x_s^2P_s^2ds \right] < \infty
	\]
	for all $t \ge 0$.
\end{definition}

\begin{assumption}[Self-financing strategies]
	\label{assum:self-financing}
	We assume that all trading strategies are \vocab{self-financing}, i.e.\ for all $t \ge 0$, if $t \in [t_n,t_{n+1})$, then
	\[
	x_tP_t + y_t - (x_0P_0 + y_0) = \int_0^t x_sdP_s + \sum_{i=1}^n\sum_{j=1}^{n_i} (\Delta x_{i,j} P_{t_i} + \Delta y_{i,j}).
	\]
	where $(\Delta x_{i,j},\Delta y_{i,j})_{j=1}^{n_i}$ are the payoffs of the on-chain trades executed by the strategy for block $i$.
\end{assumption}

Informally, a self-financing strategy does not get any assets added or removed aside from
through trades, either exchanging at the external market price or executing admissible
trades through the on-chain liquidity pool.

\begin{definition}[Sum of strategies]
	\label{def:sum_strats}
	Let $S$ be a trading strategy whose asset holdings are given by $(x_t, y_t)$ and whose
	sequence on-chain trades submitted for block $n$ is $T_n$, and
	similarly let $S'$ be a trading strategy whose asset holdings and submitted on-chain trades
	are $(x'_t, y'_t)$ and $T'_n$ respectively.
	Then their \vocab{sum} $S + S'$ is defined to be the strategy whose asset holdings are given by
	$(x_t + x'_t, y_t + y'_t)$ and whose on-chain trades submitted for block $n$ is the union
	$T_n \cup T'_n$.
\end{definition}

\begin{definition}[Arbitrage strategy]
	\label{def:arb}
	An \vocab{arbitrage strategy} is a trading strategy with $x_t = 0$ for all $t$.
\end{definition}

\begin{definition}[Simple arbitrage strategy $S_0$]
	\label{def:simple_arb}
	Given a fixed blockchain market and a fixed liquidity pool within that market,
	the \vocab{simple arbitrage strategy}, denoted by $S_0$,
	with respect to that market is defined to be the self-financing arbitrage strategy
	such that for each block $n$ with block time $t_n$
	it submits the trade $a^*(s_n,P_{t_n})$.
	\end{definition}

\begin{definition}[Concurrent strategies]
	\label{def:concurrent_trading_strategies}
	A set of trading strategies $S_1,\ldots,S_m$ is \vocab{concurrent} if their executed on-chain trades
	are disjoint.
\end{definition}

\begin{definition}[Covering strategies]
	\label{def:covering_trading_strategies}
	A set of trading strategies is \vocab{covering} if it is concurrent and the union of their on-chain trades
	equals all on-chain trades.
\end{definition}

\begin{definition}[Complete strategies]
	\label{def:complete_trading_strategies}
	A set of trading strategies is \vocab{complete} if it is covering and their sum is an arbitrage strategy.
\end{definition}

\begin{definition}[Competitive strategies]
	\label{def:competitive_trading_strategies}
	A set of $m$ trading strategies is \vocab{competitive} if it is complete and the following conditions are satisfied:
	\begin{enumerate}
		\item%
		For each block, the pool state after executing all the on-chain trades from all the strategies
		in the set is a no-arbitrage state.
		
		\item%
		either $m = 1$ or
		with probability $1$ each strategy submits non-null admissible on-chain trades for some positive number of blocks,
		and there are infinitely many blocks for which at least two of the strategies submit non-null admissible on-chain
		trades.
	\end{enumerate}	
\end{definition}

Intuitively, a concurrent set of strategies models a set of strategies that actually run in the same world and compete
for the same set of opportunities. A covering set captures all strategies that trade on-chain.
A complete set captures all strategies that trade in the external market as well.
A competitive set models the notion of a set of strategies that leave no money on the table after each block.
The second condition for competitiveness just enforces that each strategy is ``doing something'' rather than just one
strategy doing all the trading while the others sit back and watch.

\subsection{PNL and MEV}

\begin{definition}[PNL of strategy]
	\label{def:pnl}
	Let $S = (x_t,y_t)$ be a strategy. The \vocab{PNL} (\textbf{P}rofit a\textbf{n}d \textbf{L}oss) of $S$,
	denoted $\PNL(S)$, is the process
	\[
	\PNL(S)_t \triangleq x_tP_t + y_t - (x_0P_0 + y_0).
	\]
	We denote by
	\[
	\PNL^*(S)
	\]
	the PNL of $S$ if $S$ is the only trading strategy submitting on-chain trades, called the \vocab{uncontested PNL} of $S$.
\end{definition}

\begin{definition}[MEV]
	\label{def:mev}
	Let $B$ be a blockchain market and let $m$ be a positive integer.
	The \vocab{$m$-player competitive pathwise maximal extractable value (MEV)} of $B$,
	denoted $\MEV^{\pathmev}_m(B)$, is
	\[
	\MEV^{\pathmev}_m(B) \triangleq \sup_{\cS~\text{competitive}} \sum_{S \in \cS} \PNL(S)
	\]
	where $\cS$ ranges over all competitive sets of $m$ or more strategies.
	
	The \vocab{$m$-player competitive MEV} of $B$,
	denoted $\MEV_m(B)$, is
	\[
	\MEV_m(B) \triangleq \sup_{\cS~\text{competitive}} \E\left[\sum_{S \in \cS} \PNL(S)\right]
	\]
	where again $\cS$ ranges over all competitive sets of $m$ or more strategies.
	
	The \vocab{noncompetitive MEV} of $B$,
	denoted $\MEV^*(B)$, is
	\[
	\MEV^*(B) \triangleq \sup_{\cS~\text{covering}} \E\left[\sum_{S \in \cS} \PNL(S)\right]
	\]
	where $\cS$ is allowed to range over all covering sets of strategies.
\end{definition}

As we will see in \Cref{thm:comp_pmev_frictionless}, \Cref{thm:mev_frictionless}, and \Cref{thm:mev_friction},
$\MEV^{\pathmev}_m(B)$ and $\MEV_m(B)$ are independent of $m$, and so there are well-defined notions of
$\MEV^{\pathmev}(B)$ and $\MEV(B)$.

\Cref{def:mev} is worth explaining, as it may look strange at first glance.
Recall that $\PNL(S)$ is not a single number, it is a stochastic process, with a
random value at each time $t \ge 0$. Therefore $\sum_{S \in \cS} \PNL(S)$ is a stochastic process.
Since stochastic processes in general cannot be totally ordered, we must take the pathwise
supremum over each time $t$. In particular, this means that for each $t$, the set $\cS$ that approaches
the supremum may be very different, and there may not be a single $\cS$ whose total PNL
comes close to $\MEV_m(B)$. However, our first main result (\Cref{thm:comp_pmev_frictionless}) is that,
in the frictionless setting, indeed there exists a single set $\cS$ that attains this maximum value.
As for $\MEV(B)$ and $\MEV^*(B)$, these quantities are no longer random, but are still functions of time $t$.

It is important to require concurrency because otherwise one can make $\MEV_m(B)$ arbitrary large
by adding arbitrarily many copies of the same strategy, while we are interested in the value that can be
extracted by strategies competing for the same on-chain opportunities.

The covering hypothesis is also necessary to avoid a degenerate definition.
The profit of an on-chain trade can be arbitrarily large if the trade is preceded by
an arbitrarily unprofitable trade (this is the idea behind backrunning).
Therefore, we should include both trades, so that we can show that on-chain trading is constant-sum.

For pathwise MEV, completeness is required because otherwise we can make a strategy's PNL arbitrarily
large by making its risk exposure arbitrarily high in the external market.
But the external market is zero-sum, so any arbitrarily large profits made by a strategy in the
external market must be offset by corresponding losses spread across other strategies.
Completeness essentially includes those strategies being traded against.
When the price is a martingale, we do not require completeness because
any trades in the external market have zero expected return.

The competitiveness condition is the main difference between $\MEV(B)$ and $\MEV^*(B)$.
The former measures the MEV when the space is competitive to the point where every block
no profits are left on the table, while the latter measures the maximum possible value if
some profits could be deferred.

\section{Main Results}
\label{sec:main_results}

We now state our results, deferring the proofs to \Cref{sec:proofs}.

\subsection{Ordering mechanism invariance}
Our main results characterize the measures of MEV defined in \Cref{sec:formalism}
in terms of the uncontested PNL of a single strategy or a family of strategies.
This implies that the MEV is invariant to ordering mechanisms, since uncontested PNL
is invariant.

\subsubsection{Frictionless setting}

When the pool is frictionless, we show that the pathwise competitive MEV is equal to the uncontested PNL
of the simple arbitrage strategy.

\begin{theorem}[Invariance of competitive pathwise MEV, frictionless]
	\label{thm:comp_pmev_frictionless}
	Let $B$ be a blockchain market with deterministic block times whose liquidity pool is frictionless.
	Then, for all positive integers $m$, there exists a competitive set $\cS$ of $m$ strategies such that
	\[
	\MEV^{\pathmev}_m(B) = \sum_{S \in \cS} \PNL(S) = \PNL^*(S_0).
	\]
	In particular, $\MEV^{\pathmev}_m(B)$ does not depend on $m$, so the following definition of
	competitive pathwise MEV
	\[
	\MEV^{\pathmev}(B) \triangleq \MEV^{\pathmev}_m(B)
	\]
	is well-defined, and also does not depend on the ordering mechanism of $B$.
\end{theorem}

A corollary of this is that we can easily characterize the (non-pathwise) competitive MEV.

\begin{corollary}[Competitive MEV is average competitive pathwise MEV]
	\label{cor:pmev_equals_mev_frictionless}
	Let $B$ be a blockchain market with deterministic block times with a frictionless liquidity pool.
	Then
	\[
	\MEV(B) = \E[\PNL^*(S_0)] = \E[\MEV^{\pathmev}(B)].
	\]
\end{corollary}

If we additionally require that the pool is path-independent and that the asset price $P_t$ is a martingale,
then we can characterize the noncompetitive MEV as the expected uncontested PNL of the simple arbitrage strategy too,
and therefore competitive and noncompetitive MEV are equal.

\begin{theorem}[Invariance of noncompetitive MEV, frictionless]
	\label{thm:mev_frictionless}
	Let $B$ be a blockchain market with deterministic block times whose liquidity pool is frictionless and path-independent.
	Suppose $P_t$ is a martingale.
	Then
	\[
	\MEV^*(B) = \E[\PNL^*(S_0)].
	\]
	In particular, $\MEV^*(B)$ is independent of the ordering mechanism of $B$, and also
	\[
	\MEV^*(B) = \MEV(B).
	\]
\end{theorem}

\subsubsection{Fees setting}

If the pool has a fee relative to an \emph{efficient}, frictionless, path-independent pool, and the asset price
is a martingale, then we can prove analogous results. In particular, \Cref{thm:comp_mev_friction} is the analogue
of \Cref{cor:pmev_equals_mev_frictionless} and \Cref{thm:mev_friction} is the analogue of \Cref{thm:mev_frictionless}.
In \Cref{sec:counterexamples}, we give a counterexample to show how an analogue of \Cref{thm:comp_pmev_frictionless}
cannot hold in general when there are fees.

\begin{theorem}[Invariance of competitive MEV, with fees]
	\label{thm:comp_mev_friction}
	Let $B$ be a blockchain market with deterministic block times whose liquidity pool has a fee relative
	to an efficient, frictionless, path-independent pool.
	Suppose $P_t$ is a martingale.
	Then, for all positive integers $m$,
	\[
	\MEV_m(B) = \E[\PNL^*(S_0)].
	\]
	In particular, $\MEV^{\pathmev}_m(B)$ does not depend on $m$, so the following definition of competitive MEV
	\[
	\MEV(B) \triangleq \MEV_m(B)
	\]
	is well-defined, and also does not depend on the ordering mechanism of $B$.
\end{theorem}

\begin{theorem}[Invariance of noncompetitive MEV, with fees]
	\label{thm:mev_friction}
	Let $B$ be a blockchain market with deterministic block times whose liquidity pool has a fee relative
	to an efficient, frictionless, path-independent pool.
	Suppose $P_t$ is a martingale.
	For each block $n \ge 1$ with block time $t_n$, define $S_n$ to be arbitrage strategy
	that only submits an on-chain trade of $a^*(s_n,P_{t_n})$ on block $n$ and no trades
	on any other block.
	Then, for $t \ge 0$, if $t_n \le t < t_{n+1}$,
	\[
	\MEV^*(B)_t = \E[\PNL^*(S_n)].
	\]
\end{theorem}

\subsection{Dependence on block times}
We also investigate the dependence of the measures of MEV on block times, namely what happens
if we subdivide each block into multiple blocks.

In the frictionless case, if the pool is path-independent and the asset price is a martingale,
then there is no dependence on block times.

\begin{theorem}[Block time invariance of MEV, frictionless]
	\label{thm:block_times_frictionless}
	Let $B$ be a blockchain market with deterministic block times whose liquidity pool
	is frictionless and path-independent,
	and let $B'$ be identical to $B$ except each block is subdivided into $k$ blocks
	(not necessarily evenly spaced).
	Suppose $P_t$ is a martingale.
	Then, for all block times $t$ of $B$,
	\[
	\MEV(B')_t = \MEV(B)_t
	\]
	and
	\[
	\MEV^*(B')_t = \MEV^*(B)_t.
	\]
\end{theorem}

In the setting with fees, if the underlying pool is efficient, frictionless, and path-independent,
and the asset price is a martingale, then the noncompetitive MEV does not depend on block times,
but the competitive MEV may shrink as block times get shorter.

\begin{theorem}[MEV is nonincreasing with shorter block times with fees]
	\label{thm:block_times_friction}
	Let $B$ be a blockchain market with deterministic block times whose liquidity pool
	has a fee relative to an efficient, frictionless and path-independent pool,
	and let $B'$ be identical to $B$ except each block is subdivided into $k$ blocks
	(not necessarily evenly spaced).
	Suppose $P_t$ is a martingale.
	Then, for all block times $t$ of $B$,
	\[
	\MEV(B')_t \le \MEV(B)_t
	\]
	and
	\[
	\MEV^*(B')_t = \MEV^*(B)_t.
	\]
\end{theorem}

\section{Proofs}
\label{sec:proofs}
In this section we prove the main results of \Cref{sec:main_results}.

\subsection{Intermediate results}
We begin by developing some helpful intermediate results.

\subsubsection{General}
This section contains general results about strategies and PNL.
We start by observing some simple but
useful facts that follow immediately from our definitions.

\begin{proposition}
	\label{prop:sum_of_strats}
	If $S_1,\ldots,S_m$ are concurrent trading strategies, then $S_1 + \cdots + S_m$ is a trading strategy whose
	submitted on-chain trades and executed on-chain trades are the respective unions of those of $S_1,\ldots,S_m$.
\end{proposition}
\begin{proof}
	Follows immediately by induction on \Cref{def:sum_strats}.
\end{proof}

\begin{proposition}
	\label{prop:pnl_sum}
	Let $S_1,\ldots,S_m$ be trading strategies.
	Then
	\[
	\PNL(S_1 + \cdots + S_m) = \PNL(S_1) + \cdots + \PNL(S_m).
	\]
\end{proposition}
\begin{proof}
	Follows from \Cref{def:pnl} and \Cref{prop:sum_of_strats}, in particular the fact that $\PNL$ is linear
	in $x_t$, $y_t$, $\Delta x_{i,j}$, and $\Delta y_{i,j}$.
\end{proof}

The next result provides a formula for the PNL of an arbitrage strategy.

\begin{proposition}
	\label{prop:arb_pnl}
	Let $(t_n)$ be block times.
	Let $S$ be an arbitrage strategy.
	Let $(\Delta x_{i,1},\Delta y_{i,1}),\ldots,(\Delta x_{i,n_i},\Delta y_{i,n_i})$
	be payoffs of the on-chain trades executed by $S$ in block $i$.
	If $t_n \le t < t_{n+1}$, then
	\[
	\PNL(S)_t = \sum_{i=1}^n \sum_{j=1}^{n_i} (\Delta x_{i,j}P_{t_i} + \Delta y_{i,j}).
	\]
\end{proposition}
\begin{proof}
	By \Cref{def:pnl} and \Cref{assum:self-financing} we have
	\[
	\PNL(S)_t = x_tP_t + y_t - (x_0P_0 + y_0) = \int_0^t x_sdP_s + \sum_{i=1}^n\sum_{j=1}^{n_i} (\Delta x_{i,j} P_{t_i} + \Delta y_{i,j}).
	\]
	By \Cref{def:arb}, $x_s = 0$ for all $s \ge 0$, hence the integral equals $0$.
\end{proof}

The next lemma is used to show that $S_0$ is optimal in capturing MEV in a given block.

\begin{lemma}
	\label{lemma:prefix_pnl_bound}
	Fix a block with time $t$ and state $s$, and let $a_1,\ldots,a_n$ be
	an admissible sequence of trades with respect to $s$, with
	payoffs $(\Delta x_1,\Delta y_1),\ldots,(\Delta x_n,\Delta y_n)$.
	Let $P_t$ be the external market price.
	Let $a^*(s,P_t) \in A_s$ be the optimal action and let
	$\pi(a^*(s,P_t)) = (\Delta x^*,\Delta y^*)$.
	Then
	\[
	\sum_{i=1}^n (\Delta x_i P_t + \Delta y_i) \le \Delta x^* P_t + \Delta y^*.
	\]
\end{lemma}
\begin{proof}
	Let $a = a_1\cdots a_n$, which has payoff $\pi(a) = (\Delta x,\Delta y)$
	where $\Delta x = \sum_{i=1}^n x_i$ and $\Delta y = \sum_{i=1}^n y_i$.
	By the composition axiom of \Cref{assum:liq_pool_axioms}, $a$ is admissible with respect to $s$.
	By the optimality of $a^*(s,P)$,
	\[
	\sum_{i=1}^n (\Delta x_i P_t + \Delta y_i) = \Delta x P_t + \Delta y \le \Delta x^* P_t + \Delta y^*.
	\]
\end{proof}

The next lemma says that optimal actions put the pool in a no-arbitrage state.

\begin{lemma}
	\label{lemma:s0_state_after}
	For any pool state $s \in \Sigma$ and external market price $P$, $\tau(s,a^*(s,P))$
	if a no-arbitrage state relative to $P$.
	In particular, if $\Pi$ is frictionless, then
	\[
	\tau(s,a^*(s,P)) = s^*(P).
	\]
\end{lemma}
\begin{proof}
	Let $s' = \tau(s,a^*(s,P))$. Suppose there exists $a \in A_{s'}$ with $\pi(a) = (\Delta x,\Delta y)$
	such that $\Delta x P + \Delta y > 0$. Let $\pi(a^*(s,P)) = (\Delta x^*,\Delta y^*)$.
	Then the action $a^*a$ is admissible with respect to $s$ and
	\[
	\pi(a^*a) = (\Delta x^* + \Delta x, \Delta y^* + \Delta y)
	\]
	by the composition axiom of \Cref{assum:liq_pool_axioms}.
	But
	\begin{align*}
		(\Delta x^* + \Delta x)P + (\Delta y^* + \Delta y)
		&= (\Delta x^*P + \Delta y^*) + (\Delta x P + \Delta y) \\
		&> \Delta x^*P + \Delta y^*
	\end{align*}
	contradicting the optimality of $a^*(s,P)$.
	Therefore $\Delta x P + \Delta y \le 0$ for all $a \in A_{s'}$.
	Therefore $\tau(s,a^*(s,P))$ is a no-arbitrage state relative to $P$.
	
	If $\Pi$ is frictionless, then it has only one no-arbitrage state $s^*(P)$
	relative to $P$.
\end{proof}

\begin{corollary}
	\label{cor:s0_competitive}
	If the liquidity pool is frictionless, the set $\{S_0\}$ is competitive.
\end{corollary}
\begin{proof}
	Follows immediately from \Cref{lemma:s0_state_after}.
\end{proof}

\begin{corollary}
	\label{cor:no_arb_state_reachable}
	For a frictionless pool, if $s \in \Sigma$ and $P$ is some price,
	then there exists an atomic $a \in A_s \cap A$ such that $\tau(s,a) = s^*(P)$.
\end{corollary}
\begin{proof}
	Follows immediately from \Cref{lemma:s0_state_after}.
\end{proof}

The next lemma implies that concurrent copies of $S$ exactly split the pie.

\begin{lemma}
	\label{lemma:s0_mutex}
	For a given pool state $s \in \Sigma$ and external price $P$,
	if $\pi(a^*(s,P)) = (\Delta x,\Delta y)$ and $\Delta x P + \Delta y > 0$,
	and $m$ copies of $a^*(s,P)$ are submitted, then exactly one of the trades will be executed.
\end{lemma}
\begin{proof}
	Certainly at least one of the trades will be executed, so it only remains to show
	that none of the other trades are executed. It suffices to show that
	$a^*(s,P)a^*(s,P) \notin A_s$.
	By \Cref{lemma:s0_state_after}, $\tau(s,a^*(s,P)) = s^*(P)$.
	By the equilibrium state axiom in \Cref{assum:liq_pool_axioms},
	for any $a \in A_{s^*(P)}$, its payoff $(\Delta x,\Delta y)$
	satisfies $\Delta x P + \Delta y \le 0$.
	Therefore $a^*(s,P) \notin A_{s^*(P)}$, so by the composition axiom in \Cref{assum:liq_pool_axioms},
	$a^*(s,P)a^*(s,P) \notin A_s$.
\end{proof}

\begin{proposition}
	\label{prop:m_player_s0}
	For any positive integer $m$, there exists a competitive set $\cS$ of $m$ strategies such that
	\[
	\sum_{\cS \in \Sigma} \PNL(S) = \PNL^*(S_0).
	\]
\end{proposition}
\begin{proof}
	For $i \in \{1,\ldots,m\}$, let $S_i$ be a concurrent copy of $S_0$
	and let $\cS = \{S_1,\ldots,S_m\}$.
	On each block, either none of the $S_i$ submits any on-chain trade,
	or they all submit a copy of the trade $a^*(s^*(P_{t_{n-1}}), P_{t_n})$
	at block time $t_n$,
	in which case exactly one is executed (by \Cref{lemma:s0_mutex}).
	Therefore the executed trades of $S_1 + \cdots + S_m$ are exactly equal
	to that of $S_0$, hence
	\[
	\sum_{S \in \cS} \PNL(S) = \PNL(S_1 + \cdots + S_m) = \PNL^*(S_0).
	\]
	This also implies that $\cS$ is covering.
	The fact that it is complete follows from the fact that $S_0$ is an
	arbitrage strategy.
	Finally, it is competitive since it executes the same on-chain trades
	as $S_0$ which is competitive (\Cref{cor:s0_competitive}).
\end{proof}

\subsubsection{Frictionless pools}
This section contains a result that is true for frictionless pools.
It states that no competitive set of strategies can make more profit than $S_0$ by itself.

\begin{proposition}
	\label{prop:sum_of_arb_bound}
	Suppose the liquidity pool is frictionless.
	Let $S_1,\ldots,S_m$ be a competitive set of trading strategies.
	Then
	\[
	\PNL(S_1) + \cdots + \PNL(S_m) \le \PNL^*(S_0).
	\]
\end{proposition}
\begin{proof}
	Let $S = S_1 + \cdots + S_m$, which is an arbitrage strategy by \Cref{def:competitive_trading_strategies}.
	Since both sets $\{S_1,\ldots,S_m\}$ and $\{S_0\}$ are competitive (\Cref{cor:s0_competitive}), 
	under each set of concurrent strategies the pool state at block time $t_i$ is $s_i = s^*(P_{t_{i-1}})$
	since the pool is frictionless.
	Consider a fixed block at time $t_i$ and let $T$ be the on-chain trades of $S$.
	Since $S$ is complete, $T$ is an admissible sequence of trades with
	respect to $s_i$.
	Let $(\Delta x_{i,1},\Delta y_{i,1}),\ldots,(\Delta x_{i,n_i},\Delta y_{i,n_i})$ be
	the payoffs of the trades in $T$.
	Let $(\Delta x^*_i,\Delta y^*_i)$ be the payoff of $a^*(s^*(P_{t_{i-1}}),P_{t_i})$.
	Therefore, if $t_n \le t < t_{n+1}$,
	\begin{align*}
		\PNL(S_1)_t + \cdots + \PNL(S_m)_t
		&= \PNL(S)_t \tag{\Cref{prop:pnl_sum}} \\
		&= \sum_{i=1}^n\sum_{j=1}^{n_i} (\Delta x_{i,j}P_{t_i} + \Delta y_{i,j}) \tag{\Cref{prop:arb_pnl}} \\
		&\le \sum_{i=1}^n (\Delta x^*_i P_{t_i} + \Delta y^*_i) \tag{\Cref{lemma:prefix_pnl_bound}} \\
		&= \PNL^*(S_0)_t. \tag{\Cref{prop:arb_pnl}}
	\end{align*}
\end{proof}

\subsubsection{Martingales}
This section contains results that are true when $P_t$ is a martingale.
The first lemma says that if the pool is frictionless and path-independent,
and $P_t$ is a martingale, then any two strategies that end the pool in the
same state have the same expected cumulative PNL.

\begin{lemma}
	\label{lemma:martingale}
	Let $B$ be a blockchain market whose liquidity pool is frictionless and path-independent.
	Suppose $P_t$ is a martingale.
	Let $S$ be an arbitrage strategy that submits on-chain trade $a_i$ on block $i$ at time $t_i$.
	Let $S'$ be an arbitrage strategy that submits no on-chain trades except $a'$ on block $n$.
	If $\tau(s_0,a_1\cdots a_n) = \tau(s_0,a')$, then
	\[
	\E[\PNL^*(S)_{t_n}] = \E[\PNL^*(S')_{t_n}].
	\]
\end{lemma}
\begin{proof}
	Decompose $S = S_1 + S_2$ where $S_1$ submits the same on-chain trades as $S$ but has
	$x_{t_j} = \sum_{i=1}^j \Delta x_i$ for each $1 \le j \le n$ and $x_s = x_{t_j}$ for
	$t_j \le s < t_{j+1}$, where $\pi(a_i) = (\Delta x_i, \Delta y_i)$, and
	$S_2$ submits no on-chain trades and has $x_{t_j}$ offsetting the position of $S_1$,
	i.e.\ $x_{t_j} = -\sum_{i=1}^j \Delta x_i$ for $1 \le j \le n$ and $x_s = x_{t_j}$
	for $t_j \le s < t_{j+1}$.
	
	Setting $t_0 = 0$, if $x_s$ is the holdings for $x$ for $S_1$, then
	\begin{align*}
		\int_0^t x_sdP_s
		&= \sum_{j=1}^n \int_{t_{j-1}}^{t_j} x_sdP_s \\
		&= \sum_{j=1}^n x_{t_{j-1}}\int_{t_{j-1}}^{t_j} dP_s \\
		&= \sum_{j=1}^n\sum_{i=1}^{j-1} \Delta x_i (P_{t_j} - P_{t_{j-1}}) \\
		&= \sum_{i=1}^{n-1} \Delta x_i \sum_{j=i+1}^n (P_{t_j} - P_{t_{j-1}}) \\
		&= \sum_{i=1}^{n-1} \Delta x_i (P_{t_n} - P_{t_i}).
	\end{align*}
	Therefore
	\begin{align*}
		\PNL^*(S_1)_{t_n}
		&= \int_0^t x_sdP_s + \sum_{i=1}^n (\Delta x_i P_{t_i} + \Delta y_i) \\
		&= \sum_{i=1}^{n-1} \Delta x_i (P_{t_n} - P_{t_i}) + \sum_{i=1}^n (\Delta x_i P_{t_i} + \Delta y_i) \\
		&= \sum_{i=1}^n (\Delta x_i P_{t_n} + \Delta y_i) \\
		&= \left(\sum_{i=1}^n\Delta x_i\right)P_{t_n} + \left(\sum_{i=1}^n \Delta y_i\right).
	\end{align*}
	
	Let $\pi(a') = (\Delta x',\Delta y')$.
	By \Cref{prop:arb_pnl},
	\[
	\PNL^*(S')_{t_n} = \Delta x'P_{t_n} + \Delta y'.
	\]
	By hypothesis, $\tau(s_0,a_1 \cdots a_n) = \tau(s_0,a')$, so by path-independence, this implies
	\[
	\left(\sum_{i=1}^n \Delta x_i,\sum_{i=1}^n \Delta y_i\right)
	= \pi(a_1 \cdots a_n) = \pi(a')
	= (\Delta x',\Delta y').
	\]
	Therefore,
	\[
	\PNL^*(S_1)_{t_n} = \PNL^*(S')_{t_n}.
	\]
	Now, since $S_2$ exclusively trades in the external market and $P_t$ is a martingale,
	\[
	\E[\PNL^*(S_2)] = 0.
	\]
	Putting it all together, we get
	\begin{align*}
		\E[\PNL^*(S)]
		&= \E[\PNL^*(S_1) + \PNL^*(S_2)] \\
		&= \E[\PNL^*(S_1)] + \E[\PNL^*(S_2)] \\
		&= \E[\PNL^*(S'_1)]\\
		&= \E[\PNL^*(S')].
	\end{align*}
\end{proof}

The next lemma is nearly identical to \Cref{lemma:martingale} except the pool has a fee.
The conclusion is also weakened: a strategy cannot outperform (on average) a version of itself
that just waits until the end of $n$ blocks and moves the pool to the same end state.

\begin{lemma}
	\label{lemma:martingale_friction}
	Let $B$ be a blockchain market whose liquidity pool has a fee relative to a frictionless and path-independent pool.
	Suppose $P_t$ is a martingale.
	Let $S$ be an arbitrage strategy that submits on-chain trade $a_i$ on block $i$ at time $t_i$.
	Let $S'$ be an arbitrage strategy that submits no on-chain trades except $a'$ on block $n$.
	If $\tau(s_0,a_1\cdots a_n) = \tau(s_0,a')$, then
	\[
	\E[\PNL^*(S)_{t_n}] \le \E[\PNL^*(S')_{t_n}].
	\]
\end{lemma}
\begin{proof}
	Suppose the pool $\Pi$ has a fee $\phi$ relative to $\Pi_0$, which is frictionless and path-independent
	by hypothesis, and let $\pi_0$ be the payoff function of $\Pi_0$.
	Let $\PNL_0^*$ denote the uncontested PNL before fees, which is the PNL using $\pi_0$.
	For each block $i$, let $a_i$ be the trade submitted by $S$ on that block, with
	$\pi_0(a_i) = (\Delta x_i,\Delta y_i)$.
	Then, by \Cref{lemma:martingale},
	\[
	\E[\PNL_0^*(S)_t] = \E[\PNL_0^*(S')_t].
	\]
	Then,
	\begin{align*}
		\PNL^*(S')_t - \PNL^*(S)_t
		&= \PNL_0^*(S')_t - \PNL_0^*(S)_t + \phi\left(\sum_{i=1}^n |\Delta y_i| - \left|\sum_{i=1}^n \Delta y_i\right|\right) \\
		&\ge \PNL_0^*(S')_t - \PNL_0^*(S)_t
	\end{align*}
	so
	\[
	\E[\PNL^*(S')_t] - \E[\PNL^*(S)_t] \ge \E[\PNL_0^*(S')_t] - \E[\PNL_0^*(S)_t]  = 0.
	\]
\end{proof}

\subsubsection{Efficient pools}
\label{sec:efficient_pools}
We now prove some properties of efficient pools.
These results are critical in proving our main results in the setting with fees.
In particular, \Cref{cor:s0_fee_well_defined} establishes that $S_0$ is even well-defined
for a pool with a fee whose underlying pool is efficient and frictionless.
Many intuitive properties of exchanges can be proved for efficient pools.

The first result simply states that in an efficient, frictionless pool, any state is reachable
from any other state by an atomic action.

\begin{proposition}
	\label{prop:efficient_reachable}
	Let $\Pi$ be an efficient, frictionless pool.
	Let $s_1,s_2 \in \Sigma$ be two states of the pool.
	Then there exists an atomic $a \in A_{s_1} \cap A$ such that $\tau(s_1,a) = s_2$.
\end{proposition}
\begin{proof}
	Since $\Pi$ is efficient, there exists $P$ such that $s_2 = s^*(P)$.
	Apply \Cref{cor:no_arb_state_reachable}.
\end{proof}

The next result states the intuitive property that buying from the pool raises the pool price,
while selling lowers the pool price.

\begin{proposition}
	\label{prop:buy_sell_profitable}
	Let $\Pi$ be an efficient, frictionless pool.
	Let $s = s^*(P_0)$ for some price $P_0$ and let $a \in A_s$ with $\pi(a) = (\Delta x,\Delta y)$.
	Suppose $\Delta x P + \Delta y > 0$.
	\begin{itemize}
		\item%
		$P > P_0$ if and only if $\Delta y < 0$.
		\item%
		$P < P_0$ if and only if $\Delta y > 0$.
	\end{itemize}
\end{proposition}
\begin{proof}
	Since $s$ is a no-arbitrage state for $P_0$, we have $\Delta x P_0 + \Delta y \le 0$.
	Subtracting this from $\Delta x P + \Delta y > 0$ yields
	\[
	\Delta x (P - P_0) > 0.
	\]
	
	If $P > P_0$, then $\Delta x > 0$. Then we must have $\Delta y < 0$, for if $\Delta y \ge 0$,
	then $\Delta x P_0 + \Delta y > 0$, contradicting that $s$ is a no-arbitrage state for $P_0$.
	Conversely, if $\Delta y < 0$, then $\Delta x P > \Delta x P + \Delta y > 0$, hence $\Delta x > 0$, so $P - P_0 > 0$.
	
	If $P < P_0$, then $\Delta x < 0$, and so $\Delta y > \Delta x P + \Delta y > 0$.
	Conversely, if $\Delta y > 0$, then we must have $\Delta x < 0$, for if $\Delta x \ge 0$, then
	$\Delta x P_0 + \Delta y \ge \Delta y > 0$, contradicting that $s$ is a no-arbitrage state for $P_0$.
	Therefore $P - P_0 < 0$.
\end{proof}

The following result characterizes the optimal action for a pool with fee relative to an efficient, frictionless pool.
When the external price $P$ is high, the optimal action is to buy until the pool price is $\frac{P}{1+\phi}$,
while if $P$ is low, the optimal action is to sell until pool price is $\frac{P}{1-\phi}$.

\begin{proposition}
	\label{prop:competitive_fee_destination}
	Let $\Pi$ have a fee $\phi$ relative to an efficient, frictionless pool $\Pi_0$.
	If $s = s^*(P_0)$ and $P$ is some price, and there exists an atomic action $a \in A_s \cap A$ that maximizes
	\[
	\Delta x P + \Delta y - \phi|\Delta y|
	\]
	such that
	\begin{itemize}
		\item%
		if $P > P_0(1+\phi)$, then $\tau(s,a) = s^*\left(\frac{P}{1+\phi}\right)$;
		\item%
		if $P < P_0(1-\phi)$, then $\tau(s,a) = s^*\left(\frac{P}{1-\phi}\right)$;
		\item%
		otherwise, $\tau(s,a) = s$.
	\end{itemize}
	
	In particular, $s \in \Sigma$ is a no-arbitrage state (in $\Pi$) for price $P$ if and only if
	$s = s^*(P_0)$ (in $\Pi_0$) for some $\frac{P}{1+\phi} \le P_0 \le \frac{P}{1-\phi}$.
\end{proposition}
\begin{proof}
	Suppose $P > P_0(1+\phi)$. By \Cref{prop:buy_sell_profitable}, $\Delta x \frac{P}{1+\phi} + \Delta y > 0$ only if $\Delta y < 0$.
	Since we can always take $a = \bot$, this means
	maximizing $\Delta x \frac{P}{1+\phi} + \Delta y$ does change when restricting to $\Delta y < 0$,
	in which case, by multiplying the expression by $1+\phi$, it is equivalent to maximizing
	\[
	\Delta x P + \Delta y(1+\phi) = \Delta x P + \Delta y - \phi|\Delta y|.
	\]
	Since $\Pi_0$ satisfies the optimal action axiom, this maximization problem has an atomic solution $a = a^*\left(s,\frac{P}{1+\phi}\right) \in A_s \cap A$
	in $\Pi_0$, and by \Cref{lemma:s0_state_after}, $\tau(s,a) = s^*\left(\frac{P}{1+\phi}\right)$.
	
	Similarly, if $P < P_0(1-\phi)$, \Cref{prop:buy_sell_profitable} implies $\Delta x \frac{P}{1-\phi} + \Delta y > 0$ only if $\Delta y > 0$,
	so maximizing $\Delta x \frac{P}{1-\phi} + \Delta y$ does change when restricting to $\Delta y > 0$,
	in which case it is equivalent to maximizing
	\[
	\Delta x P + \Delta y(1-\phi) = \Delta x P + \Delta y - \phi|\Delta y|.
	\]
	Again, this maximization problem has an atomic solution $a = a^*\left(s,\frac{P}{1-\phi}\right) \in A_s \cap A$ in $\Pi_0$,
	and \Cref{lemma:s0_state_after} implies that $\tau(s,a) = s^*\left(\frac{P}{1-\phi}\right)$.
	
	Now suppose $P_0(1-\phi) \le P \le P_0(1+\phi)$. Suppose there exists $a \in A_s$ with $\pi(a) = (\Delta x,\Delta y)$
	and $\Delta x P + \Delta y - \phi|\Delta y| > 0$. If $\Delta y < 0$, then
	\[
	\Delta x \frac{P}{1+\phi} + \Delta y = \frac{1}{1+\phi}(\Delta x P + \Delta y(1+\phi)) = 
	\frac{1}{1+\phi}(\Delta x P + \Delta y - \phi|\Delta y|) > 0
	\]
	so \Cref{prop:buy_sell_profitable} implies that $\frac{P}{1+\phi} > P_0$, a contradiction.
	A similar argument shows that $\Delta y > 0$ cannot be true.
	Therefore $\Delta y = 0$. Since $s$ is a no-arbitrage state for $P_0$, this implies $0 \ge \Delta x P_0 + \Delta y = \Delta x P_0$,
	so $\Delta x = 0$, but this contradicts $\Delta x P + \Delta y > 0$. Hence such an action $a$ does not exist
	and the best we can do is $a = \bot$.
	
	We now move on to prove that $s \in \Sigma$ is a no-arbitrage state in $\Pi$ for price $P$ if and only if
	$s = s^*(P_0)$ in $\Pi_0$ for some $\frac{P}{1+\phi} \le P_0 \le \frac{P}{1-\phi}$.
	We have just shown the reverse direction.
	Now suppose $s$ is a no-arbitrage state, and let $s = s^*(P_0)$.
	Suppose $P_0 > \frac{P}{1-\phi}$. By what was shown earlier in this proof,
	the optimal action in $\Pi$ has payoff $(\Delta x,\Delta y)$ maximizing $\Delta x \frac{P}{1-\phi} + \Delta y$ with $\Delta y > 0$.
	Since $s$ is a no-arbitrage state in $\Pi$, this implies
	\[
	\Delta x \frac{P}{1-\phi} + \Delta y = \frac{1}{1-\phi}\left(\Delta x P + \Delta y - \phi|\Delta y|\right) = 0
	\]
	so $s$ is a no-arbitrage state in $\Pi_0$ for $\frac{P}{1-\phi}$.
	Since $\Pi_0$ is frictionless, this implies $s = s^*\left(\frac{P}{1-\phi}\right)$, and we are done.
	A similar argument shows that if $P_0 < \frac{P}{1+\phi}$, then $s = s^*\left(\frac{P}{1+\phi}\right)$ and
	again we are done.
\end{proof}

The following is the culmination of the previous results in this section and states that $S_0$ is well-defined for a
pool with fee relative to an efficient, frictionless pool.

\begin{corollary}
	\label{cor:s0_fee_well_defined}
	Let $\Pi$ have a fee relative to an efficient, frictionless pool.
	Then $\Pi$ satisfies the optimal action axiom of \Cref{assum:liq_pool_axioms},
	and in particular the simple arbitrage strategy $S_0$ is well-defined for $\Pi$.
\end{corollary}
\begin{proof}
	Let $P$ be the external market price. By \Cref{prop:competitive_fee_destination}, there exists an atomic action $a^*$
	with $\pi(a^*) = (\Delta x^*,\Delta y^*)$ maximizing
	\[
	\Delta x^* P + \Delta y^* - \phi|\Delta y^*|.
	\]
	It only remains to show that there is no action $a \in A^*$ with $\pi(a) = (\Delta x,\Delta y)$ such that
	\[
	\Delta x P + \Delta y - \phi|a| > \Delta x^* P + \Delta y^* - \phi|\Delta y^*|.
	\]
	But this follows simply from the fact that $|a| \ge |\Delta y|$ (which follows easily from the triangle inequality),
	so if such an action existed, that would imply
	\[
	\Delta x P + \Delta y - \phi|\Delta y| > \Delta x^* P + \Delta y^* - \phi|\Delta y^*|
	\]
	which contradicts the optimality of $a^*$.
\end{proof}

The following result states for an efficient, frictionless, path-independent pool one can define a nondecreasing potential function
on prices such that the potential difference between two prices represents the volume one needs to buy or sell to
change the pool price from one price to the other.

\begin{proposition}
	\label{prop:efficient_potential_function}
	Let $\Pi$ be an efficient, frictionless, path-independent pool.
	Then there exists a well-defined nondecreasing \vocab{potential function} on prices, $q:\R_+ \to \R$,
	such that for every pair of prices $P_1, P_2$, and every $a \in A_{s^*(P_1)}$ with $\pi(a) = (\Delta x,\Delta y)$,
	\[
	\tau(s^*(P_1),a) = s^*(P_2) \implies q(P_1) - q(P_2) = \Delta y.
	\]
\end{proposition}
\begin{proof}
	Pick an arbitrary price $P_0 \in \R_+$ and define $q(P_0) = 0$.
	For any price $P \in \R_+$, we define $q(P)$ as follows: by \Cref{prop:efficient_reachable},
	there exists an admissible action $a \in A_{s^*(P_0)}$ such that $\tau(s^*(P_0),a) = s^*(P)$,
	with $\pi(a) = (\Delta x,\Delta y)$; define $q(P) = -\Delta y$.
	
	Because $\Pi$ is path-independent, the definition of $q(P)$ is independent of our choice of $a$,
	and so $q$ is well-defined.
	Let $P_1, P_2$ be two prices. Let $a \in A_{s^*(P_1)}$ be an action with $\pi(a) = (\Delta x,\Delta y)$.
	Let $a_1 = a^*(s^*(P_0),P_1)$ and $a_2 = a^*(s^*(P_0),P_2)$,
	with $\pi(a_1) = (\Delta x_1,\Delta y_1)$ and $\pi(a_2) = (\Delta x_2,\Delta y_2)$.
	Then $q(P_1) - q(P_2) = \Delta y_2 - \Delta y_1$.
	If $\tau(s^*(P_1),a) = s^*(P_2)$, then
	\[
	\tau(s^*(P_0),a_2) = \tau(s^*(P_0),a_1a)
	\]
	so path-independence implies
	\[
	(\Delta x_2,\Delta y_2) = \pi(a_2) = \pi(a_1a) = (\Delta x_1 + \Delta x,\Delta y_1 + \Delta y).
	\]
	Therefore, $q(P_1) - q(P_2) = \Delta y_2 - \Delta y_1 = \Delta y$.
	
	Finally, we verify that $q$ is nondecreasing. If $P_1 < P_2$ and $a = a^*(s^*(P_1),P_2)$ with $\pi(a) = (\Delta x,\Delta y)$.
	Since $\tau(s^*(P_1),a) = s^*(P_2)$, by what we just proved above it follows that $q(P_2) - q(P_1) = -\Delta y$, so
	it suffices to show that $\Delta \le 0$.
	First suppose $\Delta x P_2 + \Delta y > 0$.
	Then \Cref{prop:buy_sell_profitable} implies $\Delta y < 0$.
	Now suppose $\Delta x P_2 + \Delta y = 0$.
	Since $a$ was optimal, this implies $s^*(P_1)$ is a no-arbitrage state for $P_2$ as well.
	Since $\Pi$ is frictionless, this implies $s^*(P_1) = s^*(P_2)$.
	So we may take $a = \bot$ with $\Delta y = 0$.
\end{proof}

The next lemma is critical to the proof of \Cref{thm:comp_mev_friction}.
It implies that $S_0$ minimizes fees paid among all competitive strategies.
However, it states something even stronger: if we $S$ is any other competitive strategy,
we can modify $S_0$ into a strategy $S_0'$ that trades like $S_0$ and on the final block $n$
makes an additional trade to get to the same state that $S$ ends on, and
even with that additional trade $S_0'$ still pays no more fees than $S$.

\begin{lemma}
	\label{lemma:competitive_fee_lower_bound}
	Let $\Pi$ have a fee $\phi$ relative to an efficient, frictionless, path-independent pool $\Pi_0$.
	Let $S$ be a competitive strategy.
	Let $n \ge 1$ and $t_n \le t < t_{n+1}$, and suppose $S$ submits trades $a_1,\ldots,a_n$ on the first $n$ blocks respectively.
	Define a strategy $S_0'$ to make the same trades as $S_0$ on blocks $1,\ldots,n-1$, i.e.\ $a^*(s_{i-1},P_{t_i})$,
	and on block $n$ submit $a^*(s_{n-1},P_{t_n})$ composed with a trade such that the state after is
	equal to $\tau(s_0,a_1\cdots a_n)$.
	If $a_1',\ldots,a_n'$ are the trades of $S_0'$, then
	\[
	\sum_{i=1}^n |a_i'| \le \sum_{i=1}^n |a_i|.
	\]
	In other words, $S$ trades at least as much volume as $S'_0$.
\end{lemma}
\begin{proof}
	Let $a_1^*,\ldots,a_n^*$ be the trades of $S_0$ on blocks $1,\ldots,n$.
	Note that $a_i' = a_i^*$ for $1 \le i \le n-1$.
	Let $P_0 \triangleq \sup\{P \mid s_0 = s^*(P)\}$.
	First, consider the case
	\[
	P_0 \le P_{t_1} \le \cdots \le P_{t_n}.
	\]
	Suppose $P_{t_n} < P_0(1+\phi)$.
	By \Cref{prop:competitive_fee_destination}, $S_0$ does not submit any non-null trades, so the result trivially follows.
	Now suppose $P_{t_n} \ge P_0(1+\phi)$.
	By \Cref{prop:competitive_fee_destination}, the only non-null trades $S_0$ submits are buys, and $S_0$ ends on state
	$s^*\left(\frac{P_{t_n}}{1+\phi}\right)$. Meanwhile, $S$ lands on $s^*(P)$ for some $P \ge \frac{P_{t_n}}{1+\phi}$.
	Let $q$ be a potential function given by \Cref{prop:efficient_potential_function}.
	Then
	\[
	\sum_{i=1}^n |a_i| - \sum_{i=1}^n |a_i^*| = q\left(\frac{P_{t_n}}{1+\phi}\right) - q(P) \ge 0.
	\]
	
	Similarly, if we redefine $P_0 \triangleq \inf\{P \mid s_0 = s^*(P)\}$ and consider the case
	\[
	P_0 \ge P_{t_1} \ge \cdots \ge P_{t_n}
	\]
	then either $P_{t_n} > P_0(1-\phi)$ in which case $S_0$ does not submit any non-null trades and the result trivially follows,
	or $P_{t_n} \le P_0(1-\phi)$ in which case there is some $P \le \frac{P_{t_n}}{1-\phi}$ such that
	\[
	\sum_{i=1}^n |a_i| - \sum_{i=1}^n |a_i^*| = q(P) - q\left(\frac{P_{t_n}}{1-\phi}\right) \ge 0.
	\]
	
	Now consider the general where $P_{t_1},\ldots,P_{t_n}$ are arbitrary. By grouping together consecutive prices moving
	in the same direction, we can assume without loss of generality that the sequence alternates directions with sufficiently
	large movements, say
	$\frac{P_{t_i}}{1+\phi} \ge \frac{P_{t_{i-1}}}{1-\phi}$ if $i$ has the same parity as $n$ and
	$\frac{P_{t_i}}{1-\phi} \le \frac{P_{t_{i-1}}}{1+\phi}$ otherwise.
	Then after block $n-1$, $S_0$ is at state $s^*\left(\frac{P_{t_{n-1}}}{1-\phi}\right)$ while
	$S$ is at state $s^*(P)$ for some $P \le \frac{P_{t_{n-1}}}{1-\phi}$, and
	\[
	\sum_{i=1}^{n-1} |a_i| - \sum_{i=1}^{n-1} |a_i^*| \ge 0.
	\]
	Since $P_{t_n} \ge \frac{P_{t_{n-1}}}{1-\phi}$, after block $n$, $S_0$ ends at state $s^*\left(\frac{P_{t_n}}{1+\phi}\right)$
	while $S$ ends at state $s^*(P')$ for some $P' \ge \frac{P_{t_n}}{1+\phi}$.
	Therefore, by \Cref{prop:efficient_potential_function}, the additional trade submitted by $S_0'$ is a buy to reach $s^*(P')$,
	and
	\begin{align*}
	|a_n'|
	&= |a_n^*| + q\left(\frac{P_{t_n}}{1+\phi}\right) - q(P') \\
	&= q\left(\frac{P_{t_{n-1}}}{1-\phi}\right) - q(P') \\
	&\le q(P) - q(P') \\
	&= |a_n|.
	\end{align*}
	Therefore,
	\[
	\sum_{i=1}^n |a_i'| = \sum_{i=1}^{n-1} |a_i^*| + |a_n'| \le \sum_{i=1}^n |a_i|.
	\]
\end{proof}

\subsection{Proofs of main results}

\subsubsection{Proof of \Cref{thm:comp_pmev_frictionless}}
\begin{proof}
	Let $\cS = \{S_1,\ldots,S_m\}$ be competitive.
	Since the pool is frictionless,
	we may apply \Cref{prop:sum_of_arb_bound} to conclude that
	\[
	\sum_{i=1}^m \PNL(S_i) \le \PNL^*(S_0).
	\]
	Taking the supremum over all complete $\cS$ yields
	\[
	\MEV^{\pathmev}_m(B) \le \PNL^*(S_0).
	\]
	
	To finish the proof, it suffices to construct a competitive set
	$\cS$ of $m$ trading strategies such that
	\[
	\sum_{S \in \cS} \PNL(S) = \PNL^*(S_0),
	\]
	since then we have
	\[
	\MEV^{\pathmev}_m(B) \ge \sum_{S \in \cS} \PNL(S) = \PNL^*(S_0).
	\]
	But this follows immediately from \Cref{prop:m_player_s0}.
\end{proof}

\subsubsection{Proof of \Cref{thm:mev_frictionless}}
\begin{proof}
	For each block $n \ge 1$, let $S_n$ be the arbitrage strategy that waits until block $n$
	and then submits $a^*(s_0,P_{t_n})$. The pool is frictionless and path-independent,
	so it has a fee of $\phi=0$ relative to itself, so we can apply
	\Cref{thm:mev_friction} to conclude that, for $t_n \le t < t_{n+1}$,
	\[
	\MEV^*(B)_t = \E[\PNL^*(S_n)_t].
	\]

	Since the pool is frictionless, it follows from \Cref{lemma:s0_state_after} that
	\[
	\tau(s_0,a^*(s_0,P_{t_1}) \cdots a^*(s_{n-1},P_{t_n}) = s^*(P_{t_n}) = \tau(s_0,a^*(s_0,P_{t_n})).
	\]
	Therefore, it follows from \Cref{lemma:martingale} that
	\[
	\E[\PNL^*(S_0)_t] = \E[\PNL^*(S_n)_t] = \MEV^*(B)_t.
	\]
\end{proof}

\subsubsection{Proof of \Cref{thm:comp_mev_friction}}
\begin{proof}
	It follows from \Cref{prop:m_player_s0} that
	\[
	\MEV_m(B) \ge \E[\PNL^*(S_0)]
	\]
	so it only remains to prove the other direction.
	Let $\cS$ be a competitive set of strategies.
	We wish to show
	\[
	\E\left[\sum_{S \in \cS} \PNL(S)\right] \le \E[\PNL^*(S_0)].
	\]
	Since $\cS$ is competitive, it is covering, so
	\[
	\PNL(S) = \PNL^*(S).
	\]
	Since $\PNL$ is additive and netting trades cannot increase fees,
	and we are trying to prove an upper bound on $\sum_{S \in \cS}\PNL(S)$, we may
	replace the set with their sum and assume without loss of generality that the
	set consists of a single strategy $S$.
	
	Suppose $\Pi$ has a fee relative to $\Pi_0$ with payoff $\pi_0$, where $\Pi_0$ is efficient.
	Let $\PNL_0^*$ denote uncontested PNL relative to $\pi_0$, i.e.\ before fees.
	Let $t_n \le t < t_{n+1}$.
	Let $a_1,\ldots,a_n$ be the trades of $S_0$.
	Define a strategy $S_0'$ to make the same trades as $S_0$ on blocks $1,\ldots,n-1$, i.e.\ $a^*(s_{i-1},P_{t_i})$,
	and on block $n$ submit $a^*(s_{n-1},P_{t_n})$ composed with a trade such that it ends in the
	state $\tau(s_0,a_1\cdots a_n)$ (there always exists such a trade by \Cref{prop:efficient_reachable}).
	Since the first $n-1$ trades of $S'_0$ and $S_0$ are identical, they pass through the same states
	in each block, and in block $n$ the strategy $S_0$ makes an optimal trade compared with $S'_0$,
	so
	\[
	\PNL^*(S'_0)_t \le \PNL^*(S_0)_t.
	\]
	Therefore, it suffices to show that
	\[
	\E[\PNL^*(S)_t] \le \E[\PNL^*(S'_0)_t].
	\]
	Let $a_1^*,\ldots,a_n^*$ be the trades of $S'_0$.
	By construction, $\tau(s_0,a_1\cdots a_n) = \tau(s_0,a_1^*\cdots a_n^*)$,
	so by \Cref{lemma:martingale} it follows that
	\[
	\E[\PNL_0^*(S)_t] = \E[\PNL_0^*(S'_0)_t].
	\]
	Since $S$ is competitive and $\Pi_0$ is efficient, it trades at least much volume as $S'_0$, hence
	\begin{align*}
		\PNL^*(S'_0)_t - \PNL^*(S)_t
		&= \PNL_0^*(S'_0)_t - \PNL_0^*(S)_t + \phi\left(\sum_{i=1}^n |a_i| - \sum_{i=1}^n |a_i^*| \right) \\
		&\ge \PNL_0^*(S'_0)_t - \PNL_0^*(S)_t \tag{\Cref{lemma:competitive_fee_lower_bound}}
	\end{align*}
	Therefore
	\[
	\E[\PNL^*(S'_0)_t] - \E[\PNL^*(S)_t] \ge \E[\PNL_0^*(S'_0)_t] - \E[\PNL_0^*(S)_t] \ge 0.
	\]
\end{proof}

\subsubsection{Proof of \Cref{thm:mev_friction}}
\begin{proof}
	By definition, $\E[\PNL^*(S_n)] \le \MEV^*(B)$, so it only remains to prove the other direction.
	Let $\cS$ be a covering set of strategies.
	Since $\PNL$ is additive and netting trades cannot increase fees,
	and we are trying to prove an upper bound on $\sum_{S \in \cS}\PNL(S)$, we may replace the set with their
	sum and assume without loss of generality that the set consists of a single strategy $S$.
	Our goal is to prove
	\[
	\E[\PNL(S)_t] \le \E[\PNL^*(S_n)_t].
	\]
	Since $S$ is covering,
	\[
	\PNL(S) = \PNL^*(S).
	\]
	Suppose the pool $\Pi$ has fee $\phi$ relative to $\Pi_0$, which is frictionless and path-independent
	by hypothesis, and let $\pi_0$ be the payoff function of $\Pi_0$.
	Let $\PNL^*_0$ denote the uncontested PNL before fees, which is the PNL using $\pi_0$.
	For each block $i$, let $a_i$ be the trade submitted by $S$ on that block, with
	$\pi_0(a_i) = (\Delta x_i,\Delta y_i)$.
	Let $\Delta x = \sum_{i=1}^n \Delta x_i$ and $\Delta y = \sum_{i=1}^n \Delta y_i$.
	Let $S'$ be an arbitrage strategy that waits until block $n$
	and submits $a_1 \cdots a_n$ for block $n$. By \Cref{lemma:martingale_friction},
	\[
	\E[\PNL^*(S)_t] \le \E[\PNL^*(S')_t].
	\]
	Let $\pi_0(a^*(s_0,P_{t_n})) = (\Delta x^*,\Delta y^*)$. By optimality of $a^*(s_0,P_{t_n})$ for $\Pi$,
	\[
	\PNL^*(S')_t = \Delta x P_{t_n} + \Delta y - \phi|\Delta y| \le \Delta x^*P_{t_n} + \Delta y^* - \phi|\Delta y^*| = \PNL^*(S_n)_t.
	\]
	Therefore,
	\[
	\E[\PNL(S)_t] = \E[\PNL^*(S)_t] \le \E[\PNL^*(S')_t] \le \E[\PNL^*(S_n)_t].
	\]
\end{proof}

\subsubsection{Proof of \Cref{thm:block_times_frictionless}}
\begin{proof}
	Let $S_0$ and $S'_0$ be the simple arbitrage strategies for $B$ and $B'$ respectively.
	In light of \Cref{thm:comp_pmev_frictionless} and \Cref{cor:pmev_equals_mev_frictionless},
	it suffices to show that
	\[
	\E[\PNL^*(S_0)_t] = \E[\PNL^*(S'_0)_t].
	\]
	Since the pool is frictionless, both $S_0$ and $S'_0$ end at the same state $s^*(P_t)$.
	Therefore we may apply \Cref{lemma:martingale} to obtain the result immediately.
\end{proof}

\subsubsection{Proof of \Cref{thm:block_times_friction}}
\begin{proof}
	We first prove the equality of $\MEV^*$.
	Let $S$ be the arbitrage strategy that waits until time $t$ and submits $a^*(s_0,P_t)$.
	By \Cref{thm:mev_friction},
	\[
	\MEV^*(B)_t = \E[\PNL^*(S)_t] = \MEV^*(B')_t.
	\]
	
	Next we prove the inequality of $\MEV$.
	Let $S_0$ and $S'_0$ be the simple arbitrage strategies for $B$ and $B'$ respectively.
	Let $S$ be an arbitrage strategy for $B$ that submits two trades on each block $n$:
	the optimal trade $a^*(s_{n-1},P_{t_n})$ and another trade $a$ to land in the same state
	as $S'_0$ (such an action $a$ always exists by \Cref{prop:efficient_reachable}).
	Since $S'_0$ is competitive, so is $S$, so by \Cref{thm:comp_mev_friction} it follows
	that
	\[
	\E[\PNL^*(S)] \le \E[\PNL^*(S_0)].
	\]
	Therefore it suffices to show that $S'_0$ is no better than $S$ on average.
	Since $S'_0$ and $S$ end at the same state at each block time of $B$, it suffices
	to show that $S'_0$ is no better than $S$ on average on each block of $B$,
	which consists of $k$ blocks of $B'$. This follows immediately from
	\Cref{lemma:martingale_friction}.
\end{proof}

\subsection{Counterexamples}
\label{sec:counterexamples}
In this section we give counterexamples to show why \Cref{thm:comp_pmev_frictionless} has
no analogue in the fee setting, and why \Cref{thm:mev_friction}, the analogue of \Cref{thm:mev_frictionless},
does not characterize $\MEV^*(B)$ with the single strategy $S_0$.

First, we consider \Cref{thm:comp_pmev_frictionless} in the fee setting.
Recall that the statement is an equality of random variables, so the equality holds
for all possible price paths.
In the fee setting, the proof breaks down because a pool with fee is no longer frictionless,
i.e.\ for a given external price $P$ it may have multiple no-arbitrage states for $P$.
Therefore, there is no guarantee that a competitive set of strategies passes through the
same sequence of states as $S_0$, so it is possible to construct price paths that reward
deviating from the path taken by $S_0$.

As a simple example, consider a liquidity pool with a linear liquidity curve:
$1$ unit per price of liquidity uniformly distributed across all prices.
If the current pool price is $p_1$, it costs $\int_{p_1}^{p_2}p\,dp = (p_2^2-p_1^2)/2$
before fees to buy $p_2 - p_1$ units which pushes the pool price to $p_2$.
If the fee is $\phi$, then the total cost would be $(p_2^2 - p_1^2)/2 - \phi|p_2-p_1|$.

Suppose both the external market price and pool price start at $P_0 = 1$ and the fee is $\phi = 0.01$.
Now suppose at the next block, the external price rises to $P_1 = 100$.
The simple arbitrage strategy $S_0$ will buy from the pool until the pool price is $100(1-\phi) = 99$,
with $(\Delta x,\Delta y) = (98, -4949)$ and therefore a profit of $4851$.
Consider an alternate competitive strategy $S_1$. It is required to push the price between
$99$ and $101$. Suppose $S_1$ buys to push the price to $101$, with $(\Delta x,\Delta y) = (100,-5151)$,
for a total profit of $4849$. Now suppose in the subsequent block the external price drops
all the way back down to $P_2 = 1$. $S_0$ starts with a pool price of $99$ and must sell until
it is $1.01$, with $(\Delta x,\Delta y) = (-97.99, 4850.99)$ and a profit of $4753$,
while $S_1$ starts with a pool price of $101$ and can also sell until it is $1.01$,
with $(\Delta x,\Delta y) = (-99.99,5048.99)$ and a profit of $4949$.
So $S_1$ loses to $S_0$ by $2$ on the first block but beats by nearly $200$ on the second block!
This is because they pushed the pool price beyond the optimal price in the first block
at a very low cost, and then when the external price dropped down $S_1$ was able to monetize
those two extra units it bought at a value of nearly $100$ each.
If the price repeatedly alternates between $1$ and $100$ or any two prices sufficiently far from
each other, then $S_1$ can outerpform $S_0$ by an arbitrarily large amount.

Next, we turn to looking at the difference between \Cref{thm:mev_frictionless} and \Cref{thm:mev_friction},
which concern noncompetitive MEV, i.e.\ how much expected value can be captured if one is free to
defer profit without worrying about others taking it.
The former essentially says that in the frictionless setting, greedily optimizing profit every block
also optimizes total profit across blocks on average, assuming the price is a martingale.
This cannot be extended to the fee setting for a simple reason: \Cref{thm:block_times_friction}.
When there are fees, it is always better to wait $n$ blocks and do one optimal trade
than to do an optimal trade at each intermediate block.

\section{Generalizations}
\label{sec:generalizations}

\subsection{Multiple liquidity pools}
\label{sec:multiple_liq_pools}
The formalism in \Cref{sec:formalism} easily handles multiple liquidity pools.
The idea is that multiple liquidity pools can be viewed as a single liquidity pool.

\begin{definition}[Product of liquidity pools]
	\label{def:liq_pool_product}
	Let
	\begin{align*}
		\Pi_1 &= (\Sigma_1, A_1, \{A_s\}_{s \in \Sigma_1}, \tau_1,\pi_1,s_{1,0})\\
		\Pi_2 &= (\Sigma_2, A_2, \{A_s\}_{s \in \Sigma_2}, \tau_2,\pi_2,s_{2,0})
	\end{align*}
	be two liquidity pools. The \vocab{product} $\Pi_1 \times \Pi_2$ of
	$\Pi_1$ and $\Pi_2$ is defined to be the liquidity pool
	$(\Sigma, A, \{A_s\}_{s \in \Sigma},\tau,\pi,s_0)$ defined by
	\begin{itemize}
		\item%
		$\Sigma = \Sigma_1 \times \Sigma_2$;
		
		\item%
		$A = A_1 \times A_2$;
		
		\item%
		If $(s_1,s_2) \in \Sigma$, then $(a_1,a_2) \in A_{(s_1,s_2)}$
		if and only if $a_1 \in A_{s_1}$ and $a_2 \in A_{s_2}$.
		
		\item%
		$\tau((s_1,s_2),(a_1,a_2)) = (\tau_1(s_1,a_1),\tau_2(s_2,a_2))$
		for every $(s_1,s_2) \in \Sigma$ and $(a_1,a_2) \in A$.
		
		\item%
		$\pi((a_1,a_2)) = \pi_1(a_1) + \pi_2(a_2)$ for every $(a_1,a_2) \in A$.
		
		\item%
		$s_0 = (s_{1,0},s_{2,0})$.
	\end{itemize}
\end{definition}

Conceptually, $\Pi_1 \times \Pi_2$ just represents having both liquidity pools side by side,
where the admissible actions are either simultaneously applying admissible actions
on each pool, or only apply an admissible action on one of the pools.
We show that the product satisfies the axioms if each of the constituents does,
which demonstrates that we may assume without loss of generality that there is just
a single liquidity pool for the asset $x$.

\begin{proposition}
	\label{prop:liq_pool_product}
	If $\Pi_1$ and $\Pi_2$ are liquidity pools satisfying the axioms of \Cref{assum:liq_pool_axioms},
	then so is their product $\Pi_1 \times \Pi_2$.
\end{proposition}
\begin{proof}
	We first verify the liquidity pool axioms for $\Pi_1 \times \Pi_2$:
	\begin{itemize}
		\item%
		(Null action):
		Let $\bot_1$ and $\bot_2$ be the respective null actions of $\Pi_1$ and $\Pi_2$.
		It is easily verified that $\bot = (\bot_1,\bot_2)$ is a null action of $\Pi_1 \times \Pi_2$.
		
		\item%
		(Composition of actions):
		If $(a_1,a_2), (a'_1,a'_2) \in A^*$, define their composition to be $(a_1a'_1,a_2a'_2)$.
		It is easily verified that this composition law satisfies the required property.
		
		\item%
		(Optimal action):
		Let $s = (s_1,s_2) \in \Sigma$ and let $P$ be the external price of asset $x$.
		Let $a^*_1 = a^*_1(s_1,P) \in A_{s_1}$ and $a^*_2 = a^*_2(s_2,P) \in A_{s_2}$ be optimal actions
		for $\Pi_1$ and $\Pi_2$ respectively, with respective payoffs $(\Delta x^*_1,\Delta y^*_1)$
		and $(\Delta x^*_2,\Delta y^*_2)$.
		Define $a^*(s,P) = (a^*_1(s_1,P),a^*_2(s_2,P)) \in A_s$, which has payoff
		$(\Delta x^*_1 + \Delta x^*_2, \Delta y^*_1 + \Delta y^*_2)$.
		Let $a = (a_1,a_2) \in A$ and suppose $\pi_1(a_1) = (\Delta x_1,\Delta y_1)$ and
		$\pi_2(a_2) = (\Delta x_2,\Delta y_2)$, so
		$\pi(a) = (\Delta x_1 + \Delta x_2,\Delta y_1 + \Delta y_2)$.
		By the optimality of $a^*_1$ and $a^*_2$,
		\begin{align*}
			\Delta x_1 P + \Delta y_1 &\le \Delta x^*_1 P + \Delta y^*_1 \\
			\Delta x_2 P + \Delta y_2 &\le \Delta x^*_2 P + \Delta y^*_2.
		\end{align*}
		Therefore
		\[
		(\Delta x_1 + \Delta x_2)P + (\Delta y_1 + \Delta y_2)
		\le (\Delta x^*_1 + \Delta x^*_2)P + (\Delta y^*_1 + \Delta y^*_2)
		\]
		so $a$ is optimal.
	\end{itemize}
	This completes the verification that $\Pi_1 \times \Pi_2$ satisfies the liquidity pool axioms.
\end{proof}

The next two results show that the frictionless and path-independent properties are preserved by the product operation.

\begin{proposition}
	\label{prop:liq_pool_product_frictionless}
	If $\Pi_1$ and $\Pi_2$ are frictionless, then so is $\Pi_1 \times \Pi_2$.
\end{proposition}
\begin{proof}
	Let $P$ be the external price of the asset $x$.
		Let $s^*_1(P)$ and $s^*_2(P)$ be the no-arbitrage states of $\Pi_1$ and $\Pi_2$ respectively.
		Define $s^*(P) = (s^*_1(P),s^*_2(P))$.
		Suppose $a = (a_1,a_2) \in A_{s^*(P)}$ where
		$\pi_1(a_1) = (\Delta x_1,\Delta y_1)$ and $\pi_2(a_2) = (\Delta x_2,\Delta y_2)$.
		Since $s^*_1(P)$ and $s^*_2(P)$ are no-arbitrage states,
		\[
		(\Delta x_1 + \Delta x_2)P + (\Delta y_1 + \Delta y_2)
		= (\Delta x_1 P + \Delta y_1) + (\Delta x_2 P + \Delta y_2) \le 0,
		\]
		therefore $s^*(P)$ is a no-arbitrage state.
		It only remains to show that $s^*(P)$ is unique.
		Let $s = (s_1,s_2) \in \Sigma$ be a different state, $s \ne s^*(P)$.
		Without loss of generality, suppose $s_1 \ne s^*_1(P)$.
		Then $s_1$ is not a no-arbitrage state since $s^*_1(P)$ is unique,
		so there exists $a_1 \in A_{s_1}$ with $\pi_1(a) = (\Delta x,\Delta y)$ such that
		$\Delta x P + \Delta y > 0$.
		Then the action $a = (a_1,\bot_2) \in A_s$ has $\pi(a) = (\Delta x,\Delta y)$ and
		therefore $\Delta x P + \Delta y > 0$.
		Therefore $s$ is not a no-arbitrage state.
\end{proof}

\begin{proposition}
	\label{prop:liq_pool_product_path_independent}
	If $\Pi_1$ and $\Pi_2$ are path-independent, then so is $\Pi_1 \times \Pi_2$.
\end{proposition}
\begin{proof}
	Let $s = (s_1,s_2) \in \Sigma$.
	Suppose $a = (a_1,a_2)$ and $a' = (a'_1,a'_2)$ are actions such that
	$\tau(s,a) = \tau(s,a')$. Unrolling \Cref{def:liq_pool_product}, this means
	\[
	(\tau_1(s_1,a_1),\tau_2(s_2,a_2)) = \tau(s,a)
	= \tau(s,a') = (\tau_1(s_1,a'_1),\tau_2(s_2,a'_2))
	\]
	so $\tau_1(s_1,a_1) = \tau_1(s_1,a'_1)$ and $\tau_2(s_2,a_2) = \tau_2(s_2,a'_2)$.
	Since $\Pi_1$ and $\Pi_2$ are path-independent, this implies that
	$\pi_1(a_1) = \pi_1(a'_1)$ and $\pi_2(a_2) = \pi_2(a'_2)$.
	Therefore,
	\[
	\pi(a) = \pi_1(a_1) + \pi_2(a_2) = \pi_1(a'_1) + \pi_2(a'_2) = \pi(a').
	\]
\end{proof}

It follows from \Cref{prop:liq_pool_product_frictionless} and \Cref{prop:liq_pool_product_path_independent}
that any statement that holds for frictionless or path-independent pools holds for products of such pools,
and in particular \Cref{thm:comp_pmev_frictionless}, \Cref{cor:pmev_equals_mev_frictionless},
\Cref{thm:mev_frictionless}, and \Cref{thm:block_times_frictionless} all hold for multiple pools.

Extending the results in the fee setting is less automatic.
The product of two efficient pools is not necessarily efficient.
In fact, it almost certainly is not, because cross-exchange arbitrages are possible.
Moreover, mutiple pools may have distinct fee rates. However, our results still hold for
products of pools with fees relative to efficient, frictionless, path-independent pools.
It follows from the following definition and result.

\begin{definition}[Product of trading strategies]
	\label{def:product_strategy}
	Let $\Pi_1$ and $\Pi_2$ be liquidity pools, and let $S_1$ and $S_2$ be trading strategies
	that trade on $\Pi_1$ and $\Pi_2$ respectively.
	The \vocab{product} of $S_1$ and $S_2$, denoted by $S_1 \times S_2$, is the strategy
	whose position is given by the sum of the positions and on each block it submits
	$(a_1,b_1),\ldots,(a_n,b_n)$, if $S_1$ submits $a_1,\ldots,a_n$ and $S_2$ submits
	$b_1,\ldots,b_n$ (we may assume they submit the same number of trades by padding
	with $\bot$).
\end{definition}

\begin{proposition}
	\label{prop:product_strategy}
	Let $\Pi_1$ and $\Pi_2$ be liquidity pools with product $\Pi = \Pi_1 \times \Pi_2$.
	Let $S$ be a trading strategy trading on $\Pi$.
	Then $S$ can be decomposed into a product of strategies $S = S_1 \times S_2$
	such that $S_1$ trades on $\Pi_1$ and $S_2$ trades on $\Pi_2$, and
	\[
	\PNL(S) = \PNL(S_1) + \PNL(S_2).
	\]
\end{proposition}
\begin{proof}
	For any two strategies $S_1$ and $S_2$, it follows from $\pi((a_1,a_2)) = \pi_1(a_1) + \pi_2(a_2)$ that
	\[
	\PNL(S_1 \times S_2) = \PNL(S_1) + \PNL(S_2).
	\]
	So it only remains to show that we can decompose $S$ into a product.
	Each on-chain trade submitted by $S$ takes the form $(a_1,a_2) \in A_1^* \times A_2^*$,
	For each such trade submitted by $S$, we define $S_1$ to submit the trade $a_1$ and $S_2$ to submit the trade $a_2$.
	For the positions of $S_1$ and $S_2$, we can arbitrarily split the position of $S$.
	Then $S = S_1 \times S_2$.
\end{proof}

Each pool with fee satisfies the liquidity pool axioms, by \Cref{cor:s0_fee_well_defined},
so by \Cref{prop:liq_pool_product} their product does as well.
Therefore the simple arbitrage strategy $S_0$ is well-defined for the product.
From \Cref{prop:product_strategy}, it follows that $S_0$ can be decomposed into strategies
on each pool in the product, and moreover from the proof of the decomposition and the
definition of the product's optimal action, we see that $S_0$ can be decomposed into
a product of simple arbitrage strategies on each pool.
Furthermore, it follows from \Cref{prop:product_strategy} that
\[
\MEV(B_{\Pi_1 \times \Pi_2}) = \MEV(B_{\Pi_1}) + \MEV(B_{\Pi_2})
\]
and similarly for $\MEV^*$. Therefore \Cref{thm:comp_mev_friction}, \Cref{thm:mev_friction}, and \Cref{thm:block_times_friction}
extend to a product of pools each with a fee relative to some efficient, frictionless, path-independent pool.

\subsection{Multiple risky assets}
\label{sec:multiple_assets}
To extend our results to multiple assets, we must generalize our definitions and verify that our intermediate results still hold
or are even well-defined.
If we have $n$ assets, none of which are the \numeraire{}, we denote the position by a vector $x \in \R^n$
and external market prices by a vector $P \in \R_+^n$.
Then all instances of $xP + y$ generalize to $x^\top P + y$.
Payoff functions generalize to $\pi: A^* \to \R^{n+1}$.
The results in the frictionless setting easily follow in the general setting.

It is less clear how to generalize the definitions of volume (\Cref{def:volume})
and therefore fees (\Cref{def:liq_pool_fees}), as well as the results on efficient pools (\Cref{sec:efficient_pools}),
which assume the payoff of each action only involves two assets, one of which is the \numeraire{} $y$,
and therefore the results in the fee setting.

The definitions and results can be generalized as follows.
We (arbitrarily) order the assets
\[
x_1 \prec \cdots \prec x_n \prec x_{n+1} = y.
\]
We assume that each efficient pool only involves two assets $x_i,x_j$ where $i < j$.
We define the volume of an action $a \in A^*$ to be a vector $|a| \in \R^{n+1}$ as follows:
If $a \in A$ is atomic, then
\[
\pi(a) = (0,\ldots,\Delta x_i,0,\ldots,\Delta x_j,0,\ldots,0) \implies |a| = (0,\ldots,|\Delta x_j|,0,\ldots,0)
\]
where $\pi(a)$ has entries $\Delta x_i$ and $\Delta x_j$ in the $i$th and $j$th coordinates respectively and zero elsewhere,
and $|a|$ has $|\Delta x_j|$ in the $j$th coordinate and zero elsewhere.
If $a = a_1a_2$ is a composition, then $|a| = |a_1| + |a_2|$ as before.
Then we may define a pool with fee even more succinctly than before with
\[
\pi(a) = \pi_0(a) - \phi|a|.
\]
Then all of the results in \Cref{sec:efficient_pools} still hold if we replace
external market price $P$ with the ratio of external market prices $P_i/P_j$.

\section{Examples of liquidity pools}
\label{sec:example_pools}
In this section we show that a wide range of AMMs satisfy the liquidity pool axioms of \Cref{assum:liq_pool_axioms},
and furthermore have frictionless underlying pools that are also path-independent and efficient.

\subsection{Constant function market makers}
We show that under mild assumptions about the invariant, a CFMM without fees satisfies the liquidity pool axioms
and is frictionless, path-independent, and efficient.
Suppose the invariant of the CFMM is given by
\[
f(x,y) = L
\]
for some fixed $L \in \R$.
This defines a curve $\Sigma \subseteq \R_+^2$ which can be the state space.
The atomic actions are $(\Delta x,\Delta y)$, with transition function
\[
\tau((x,y), (\Delta x,\Delta y)) = (x - \Delta x, y - \Delta x).
\]
The action $(\Delta x,\Delta y)$ is admissible with respect to $(x,y)$ if and only if
\[
f(x-\Delta x,y - \Delta y) = f(x,y).
\]
The payoff function $\pi$ is just the identity on $\R^2$.

We can model this as a liquidity pool as follows.
The state space is simply $\R_+^2$, with the asset reserves $(x,y)$ as the state.
The action space is simply $\R^2$, with each action $(\Delta x,\Delta y)$ representing
the transfer of assets to the trader, with the payoff function $\pi$ being the identity function.
The transition function is $\tau((x,y),(\Delta x,\Delta y)) = (x - \Delta x, y - \Delta y)$.
If we wish to model a trading fee, then the transition function would be
$\tau((x,y),(\Delta x,\Delta y)) = (x - \Delta x', y - \Delta y')$
where $\Delta x'$ and $\Delta y'$ represent the amounts with the fee added in the
appropriate asset depending on which quantity is negative.
Then $(\Delta x,\Delta y)$ is admissible if and only if
$f(x - \Delta x', y - \Delta y') = L$.

The null action axiom is satisfied by $(0,0)$.
Composition of actions is defined by addition and can be seen to satisfy
the composition axiom.
We will show the optimal action exists by first showing that such a pool is frictionless.

A state $(x,y)$ is a no-arbitrage state for external price $P$ if and only if
the line with slope $-P$ through $(x,y)$ is tangent to $\Sigma$.
Suppose the gradient $\nabla f = (\partial_x f,\partial_y f)$ is defined and continuous
almost everywhere on $\Sigma$.
Then the vector $(-\partial_y f,\partial_x f)$ is tangent to $\Sigma$ at $(x,y)$
and so $(x,y)$ is a no-arbitrage state if $-\frac{\partial_x f}{\partial_y f} = P$.
The optimal action is the one that brings the state to the state that satisfies this.
If $\frac{\partial_x f}{\partial_y f}$ restricted to $\Sigma$ is a bijection with $\R_+$,
then the pool is frictionless and efficient. Path-independence is easily verified.

An application of this is to constant product market makers such as Uniswap V2.
For constant product markets, $f(x,y) = xy$, so $(x,y)$ is a no-arbitrage state for $P$
if and only if $y/x = P$.

The StableSwap invariant implemented by Curve v1 stablecoin pools (\cite{CurveV1})
can also be shown to satisfy these properties, being an interpolation between
constant sum and constant product markets.

Note that constant sum markets themselves do not satisfy the liquidity axioms.

\subsection{Uniswap v3}
Uniswap v3 (\cite{UniV3}) offers an innovation over constant product market makers like Uniswap v2
and CFMMs in general by supporting \emph{concentrated liquidity}, allowing liquidity providers
to specify a price band in which to concentrate their liquidity, as opposed to spreading it
out across the entire price range of $(0,\infty)$ as in CFMMs.
This allows for greater capital efficiency, and as a result Uniswap v3 pools
are the most actively traded AMMs on Ethereum and DeFi in general.

The state space of Uniswap V3 is parameterized by the pool price $p$ and consists of
several translated copies of Uniswap V2 states stitched together, each copy
representing a price band $[p_a,p_b)$.
An action either takes the state to another point on the same Uniswap V2 curve
or crosses into another band.
One can verify that Uniswap V3 without fees satisfies the liquidity pool axioms
and is frictionless, path-independent, and efficient.

\section{Discussion and Open Questions}
We have shown that an risk-neutral arbitrage-centric measure of MEV is invariant to some
classes of transformations to a blockchain protocol design, namely ordering
mechanisms and block times, assuming the liquidity resides in protocols satisfying
our liquidity pool axioms, that block times are deterministic, and that
liquidity providers are passive.

One consequence of this theoretical result is that certain ordering mechanisms,
such as giving a trader guaranteed priority in a block or even exclusive access,
do not change the overall profit that can be extracted, even if coupled with
mechanisms that effectively decrease the time between blocks for the trader.
Another corollary is that even given such a benefit, the simple arbitrage strategy
optimally extracts the value of that benefit.

An interesting line of further work is to consider measures of value other than
pure risk-neutral dollar profit. For example, traders may have different risk
preferences, and different ordering mechanisms or block time distributions
may increase the overall utility if not the overall dollar profit amount.

Another interesting direction is to extend our formalism to include DEXs
that do not satisfy our current liquidity pool axioms, even in the absence
of liquidity provider actions, e.g.\ pools whose state changes from
price oracles, whether external or internal, such as Curve v2 crypto pools (\cite{CurveV2}).

Another question that is interesting purely theoretically is
the invariance of MEV when block times are not deterministic, but follow a
Poisson point process as in blockchains that use a PoW consensus protocol.

\section{Acknowledgements}
We thank Lucas Baker, Jeff Bezaire, Christopher Chung, Jennifer Pan, and
Michael Setrin, Nihar Shah, and Anirudh Suresh for helpful conversations,
inspiration, and feedback.

\bibliographystyle{alpha}
\newcommand{\etalchar}[1]{$^{#1}$}

\end{document}